\documentclass[prx,a4paper,aps,twocolumn,superscriptaddress,longbibliography]{revtex4-1}

\usepackage{amsmath,amsthm,amssymb,amsfonts,setspace}
\usepackage{bm}
\usepackage{float}
\usepackage{color}
\usepackage[usenames,dvipsnames]{xcolor}
\usepackage[colorlinks]{hyperref}
\usepackage{mathtools}
\usepackage{cleveref} 

\usepackage[normalem]{ulem} 

\usepackage{graphicx}
\graphicspath{ {./figs/}} 

\crefformat{equation}{Eq.~(#2#1#3)} 
\crefformat{section}{Section~#2#1#3}
\crefformat{theorem}{Theorem~#2#1#3}
\crefformat{lemma}{Lemma~#2#1#3}
\crefformat{appendix}{Appendix~#2#1#3}
\Crefformat{equation}{Equation~(#2#1#3)}
\crefformat{figure}{Fig.~#2#1#3}
\crefformat{corollary}{Corollary~#2#1#3}
\crefrangeformat{equation}{Eqs.~#3(#1)#4--#5(#2)#6}

 \theoremstyle{plain}
 
 \theoremstyle{plain}
 
 \theoremstyle{plain}
 \newtheorem{thm}{Theorem}

 \theoremstyle{plain}
 \newtheorem{corr}{Corollary}
 \theoremstyle{plain}
 
 \theoremstyle{remark}
 \newtheorem*{rem*}{Remark}
 \theoremstyle{plain}

\theoremstyle{plain}

\newcommand{\ket}[1]{\vert{#1}\rangle} 
 \newcommand{\ketbra}[2]{\left\vert{#1}\right\rangle\!\left\langle{#2}\right\vert}
 \newcommand{\braket}[2]{\left\langle{#1}\middle\vert  {#2} \right\rangle}

\newcommand{\Tr}{\mathrm{Tr}}
\newcommand{\tr}{\mathrm{tr}}

\newcommand{\U}{\mathrm{U}} 

\newcommand{\eqdef}{\mathrel{:=}}

\theoremstyle{plain}

\theoremstyle{plain}
\theoremstyle{plain}
\newtheorem{lemma}{Lemma}
\theoremstyle{plain}

\theoremstyle{plain}
\newtheorem{remark}{Remark}
\theoremstyle{plain}

\begin{document}





\title{Measuring relational information between quantum states, and applications}

\author{Micha\l\ Oszmaniec}
\email{oszmaniec@cft.edu.pl}
\affiliation{ Center for Theoretical Physics, Polish Academy of Sciences, Al. Lotnik\'{o}w 32/46, 02-668 Warsaw, Poland}

\author{Daniel J.\ Brod}
\email{danieljostbrod@id.uff.br}
\affiliation{Instituto de F\'{\i}sica, Universidade Federal Fluminense, Niter\'oi, RJ, 24210-340, Brazil}

\author{Ernesto F. Galv\~ao}
\email{ernesto.galvao@inl.int}
\affiliation{International Iberian Nanotechnology Laboratory (INL), Av. Mestre Jose Veiga, 4715-330, Braga, Portugal}
\affiliation{Instituto de F\'{\i}sica, Universidade Federal Fluminense, Niter\'oi, RJ, 24210-340, Brazil}

\date{\today}

\begin{abstract}

The geometrical arrangement of a set of quantum states can be completely characterized using relational information only. This information is encoded in the pairwise state overlaps, as well as in Bargmann invariants of higher degree written as traces of products of density matrices. We describe how to measure Bargmann invariants using suitable generalizations of the SWAP test. This allows for a complete and robust characterization of the projective-unitary invariant properties of any set of pure or mixed states. As applications, we describe basis-independent tests for linear independence, coherence, and imaginarity. We also show that Bargmann invariants can be used to characterize multi-photon indistinguishability.
\end{abstract}

\maketitle

\textit{Introduction.} The relative geometrical arrangement of a set of quantum states determines properties such as the Hilbert space dimension they span  or whether they can be simultaneously diagonalized. These properties are physically meaningful, being independent of gauge choices for the global phase of each state's wave-function, and invariant under application of the same unitary transformation on all quantum states in the set.  More generally, we use the terms \emph{relational} or \emph{Projective-Unitary} (PU)  invariants to refer to properties of tuples of quantum states that are invariant under application of the same unitary transformation to each state in the tuple. 
The simplest example  of a relational invariant is the two-state overlap $\Delta_{12}=\tr(\rho_1\rho_2)$, which for pure states reduces to $\Delta_{12}=\left\lvert\left\langle \psi_1 | \psi_2 \right\rangle\right\lvert^2$, with a clear operational physical interpretation.

One way to fully characterize relational information in a tuple of (generally mixed) quantum states, $\mathcal{S}=(\rho_i)_{i=1}^N$, is to perform individual tomographic reconstruction of each state. This is both experimentally costly and unnecessary, providing much more information than strictly needed. A more economical approach is to directly measure a set of 
invariant properties of a tuple of quantum states. We call such a set \emph{complete} if it allows us to decide whether two tuples of quantum states $\mathcal{S}=(\rho_i)_{i=1}^N, \mathcal{S'}=(\sigma_i)_{i=1}^N$ are unitarily equivalent, in the sense that there exists a unitary operator $U$ taking states from $\mathcal{S}$ onto their counterparts in $\mathcal{S}'$: $\sigma_i=U\rho_i U^\dagger $. In the above, we identified the relational properties of quantum states with equivalence classes (orbits) of the unitary group acting on tuples of quantum states. This approach follows the spirit of the Erlangen program in Geometry \cite{Erlangen1900} (i.e. studying geometry via the lens of group actions and their invariants) and is justified by the fact that unitary channels constitute the most general inevitable operations allowed in quantum theory. Note that the same paradigm motivated the usage of invariant polynomials in the context of classification of entanglement classes subject to local unitary transformations \cite{LeiferLW04}.


In this Letter we introduce quantum circuits we call \textit{cycle tests}, which enable the direct measurement of complete sets of Bargmann invariants for both mixed and pure quantum states. We review previous mathematical literature identifying complete sets of invariants for different scenarios. For the case of pure states, we show all necessary invariants can be incorporated in a single Gram matrix $G_{ij}=\left\langle \psi_i |\psi_j \right\rangle$, with suitable gauge choices so that all parameters in $G$ are PU-invariants measurable by our proposed circuits. This operational approach to  measuring and using  Bargmann invariants results in applications such as multi-photon indistinguishability certification, and basis-independent tests of linear independence, imaginarity, and coherence.

We start by reviewing the definition of Bargmann invariants. Consider a tuple of $m$ pure quantum states (understood as rank one normalized projectors) $\Psi = (\psi_1, \psi_2, \ldots \psi_m)$, where we used the shorthand $\psi_j=\ketbra{\psi_j}{\psi_j}$. We later consider the generalization to mixed states. The Bargmann invariant of this tuple of states is defined as \cite{Bargmann64, ChienW16}:
\begin{equation}
\Delta_{12\cdots m}=\left\langle\psi_1|\psi_2 \right\rangle \left\langle\psi_2|\psi_3\right\rangle \cdots \left\langle\psi_{m}|\psi_1\right\rangle.
\end{equation}
We say that this invariant has \emph{degree} $m$.
Due to the suitable arrangement of bras and kets, the above expression is a well-defined PU-invariant function of states in $\Psi$. Note that, for $m=2$, this invariant reduces to the overlap $\Delta_{12}=\left|\left\langle\psi_1|{\psi_2}\right\rangle\right|^2$. Bargmann invariants have been studied in the context of characterization of photonic indistinguishability \cite{Menssen2017,Jones2020,Minke21} and geometric phases \cite{SimonM93, Chruscinski04}.
The phase of $\Delta_{12\cdots m}$ is precisely the Pancharatnam geometric phase \cite{Pancharatnam56} acquired by a wavefunction $\ket{\psi_1}$ subjected to a sequence of projective measurements onto states $\psi_m,..., \psi_2,\psi_1$. Equivalently, it is the Berry phase acquired by a cyclic trajectory along the shortest (geodesic) cyclic path visiting those states \cite{Samuel88} -- for a review of these ideas, see \cite{Chruscinski04}.



\textit{Measuring Bargmann invariants.} Let us now describe a quantum circuit family that measures any Bargmann invariant $\Delta_{12\cdots m}$. We call these circuits \textit{cycle tests}, and they are  generalizations of the well-known SWAP test used to measure the two-state overlap $\Delta_{12}$ (see Fig. \ref{fig:swap}-a).  The key gate in the cycle test circuit of Fig. \ref{fig:swap}-b is a controlled cycle permutation, which implements the cyclic permutation $(12\cdots m)$ of the input states when the control qubit is in state $\ket{1}$, and identity otherwise.  It is easy to check that the final measurement of the auxiliary qubit gives an estimate of  $\Re(\Delta_{12\cdots m})$. The circuit of Fig. \ref{fig:swap}-a was proposed in \cite{Ekert02} as a way to measure non-linear functionals of a single density matrix. Here, we use it with different inputs to estimate $\Re(\Delta_{12\cdots m})$, and propose a modification with an added phase gate $P=diag(1,i)$, to enable the  measurement of $\Im(\Delta_{12\cdots m})$ (see Fig. \ref{fig:swap}-b). Note that the cycle test reduces to the SWAP test for $m=2$ states. Overlaps can also be evaluated as a particular case of degree-3 invariants, as  $\Delta_{iij}=\Delta_{ij}=\left|\left\langle \psi_i |\psi_j \right \rangle \right|^2$ --- lower-degree invariants can always be recovered from higher-degree invariants with repeated indices.
The cycle test circuit can be implemented in linear depth using nearest-neighbor gates only, or in a depth that scales like $O(\log m)$ if we assume all-to-all qubit connectivity, as we describe in Appendix \ref{app:cycle}.



\begin{figure}[t]
    \centering
    \includegraphics[width=0.3\textwidth]{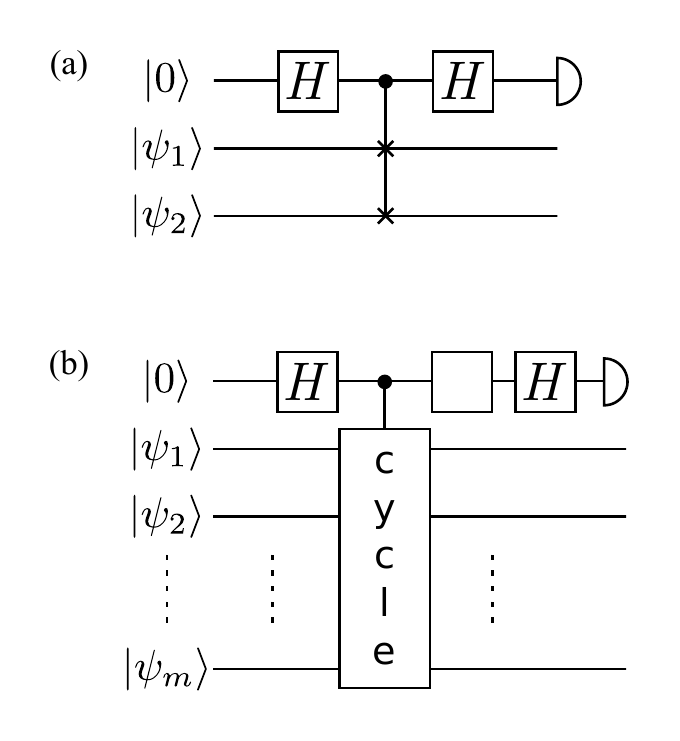}
    \caption{Circuits for measuring projective-unitary invariants of a tuple of states. (a) SWAP test to measure the two-state overlap $\Delta_{12}=\left|\left\langle\psi_1 |{\psi_2}\right\rangle\right|^2$. The probability of outcome 0 is $p(0)=(1+\Delta_{12})/2$. (b) Cycle test to measure the $m$-state invariant $\Delta_{12\cdots m}$. If the white square is the identity, we have $p(0)=(1+\Re(\Delta_{12\cdots m}))/2$. If the white square is the $P = \textrm{diag}(1,i)$ gate, we have $p(0)=(1+\Im(\Delta_{12\cdots m}))/2$. The controlled gate implements the cyclic permutation $(123\cdots m)$ of the inputs if the control is in state $\ket{1}$, and does nothing otherwise.}
    \label{fig:swap}
\end{figure}


\textit{Complete characterization of relational information between states.} The following Theorem gives a complete characterization of PU-invariant properties of a collection of $N$ pure sates in terms of Bargmann invariants.

\begin{thm}[Characterization of relational properties of pure states]\label{thm:PureSates}
Let $\Psi= (\psi_i)_{i=1}^N$ be a tuple of pure quantum states. Then, the PU-equivalence class of $\Psi$ is uniquely specified by values of at most $(N-1)^2$ Bargmann invariants. The invariants are of degree $m\leq N$ and their choice depends on $\Psi$.
\end{thm}

The above Theorem follows from results obtained by Chien and Waldron in \cite{ChienW16}. In what follows we present a simplified and shortened proof. Our strategy is based on encoding complete PU-invariants in a single Gram matrix in a way that depends on orthogonality relations of states in $\Psi$. This allows us to upper bound the number of necessary invariants by $(N-1)^2$.  and employ the results to the problem of approximate PU equivalence of quantum states. 
\begin{proof}
We start with the connection between PU equivalence of two tuples of pure states  $\Psi= ( \psi_i)_{i=1}^N$, $\Phi= ( \phi_i)_{i=1}^N$ and unitary equivalence between the associated tuples of wave-functions. Namely, $\Psi$ is PU equivalent to $\Phi$ if and only if it is possible to find representing wave functions $\mathcal{A}=\left(\ket{\psi_i} \right)_{i=1}^N $,
$\mathcal{B}=\left(\ket{\phi_i} \right)_{i=1}^N$ that are unitarily equivalent. That is, there exists an unitary operator $U$ such that $\ket{\phi_i}=U\ket{\psi_i}$, for $i=1,\ldots,N$. The problem of unitary equivalence of tuples of vectors is equivalent to equality of the corresponding Gram matrices i.e.,
$G^{\mathcal A}_{ij}=\braket{\psi_i}{\psi_j}=\braket{\phi_i}{\phi_j}=G^{\mathcal B}_{ij}$ 
(this follows e.g. from the uniqueness the Cholesky decomposition of $G$ \cite{johnston2021advanced}). The phase of individual wave functions 
$\ket{\psi_i}$ is an unphysical, gauge degree of freedom, and therefore the Gram matrix of a collection of pure states $\Psi $ is uniquely defined only up to conjugation via a diagonal unitary matrix. 
Assume now that for every tuple of quantum states $\Psi= ( \psi_i)_{i=1}^N$ we have a construction (presented in the next paragraph) of a valid  Gram matrix $\tilde{G}^\Psi$  whose entries can be expressed solely in terms of PU-invariants of states from $\Psi$. It then follows from the above considerations that $\Psi$ is PU equivalent to $\Phi$ if and only if $\tilde{G}^\Psi=\tilde{G}^\Phi$.  

The construction of $\tilde{G}^\Psi$ proceeds as follows. We first introduce the frame graph $\Gamma^\Psi$ as an (undirected) graph whose vertices are labeled by $i=1,\ldots,N$ with edges connecting $i$ and $j$ if and only if $\tr(\psi_i\psi_j)\neq 0$. Without loss of generality we assume that $\Gamma^\Psi$ is connected \footnote{If $G^\Psi$ is not connected then the problem of PU equivalence of states from $\Psi$ reduces to independent problems concerning states associated with connected components in $\Gamma^\Psi$.} i.e., that every pair of vertices in $\Gamma^\Psi$ can be connected via a path in $\Gamma^\Psi$. We can now choose a subgraph of $\Gamma^\Psi$, denoted by $\mathcal{T}^\Psi$, which is connected and contains the same vertices as $\Gamma^\Psi$, but no cycles. This is known as a \emph{spanning tree} of $\Gamma^\Psi$, and there exists at least one spanning tree for any connected graph \cite{gross1998graph} -- see Fig.\ \ref{fig:concept} for an illustration. We now choose vector representatives $\ket{\psi_i}$ of states in $\Psi$ in such a way that, for $\{i,j\}$ an edge in $\mathcal{T}^\Psi$, $\braket{\psi_i}{\psi_j}=\left|\braket{\psi_i}{\psi_j}\right|$. Every other inner product $\braket{\psi_i}{\psi_j}$ will be either 0, or its phase will be fixed as follows. We first find a path from $j$ to $i$ within $\mathcal{T}^\Psi$. Since $\mathcal{T}^\Psi$ is a spanning tree, such a path is guaranteed to exist and to be unique. Suppose this path has $k$ vertices $(\alpha_1=j, \alpha_2, \dots , \alpha_{k-1}, \alpha_k=i)$. Consider now the $k$-cycle that would be formed by adding the vertex $j$ at the end of this path (note that vertices $i$ and $j$ are connected by an edge in $\Gamma^\Psi$) 
and denote it by $C_{ij}$. 
By construction, every edge in $C_{ij}$ except for $\{i,j\}$ is in $\mathcal{T}^\Psi$, and therefore all the inner products
$\braket{\psi_{\alpha_j}}{\psi_{\alpha_{j+1}}}$ for $j \in \{1,2,\dots k \}$ have been chosen as positive. Hence, if we denote the degree-$k$ Bargmann invariant associated to $C_{ij}$ as $\Delta(C_{ij}):=
 \Delta_{\alpha_1 \alpha_2 \ldots \alpha_{k-1} \alpha_k}$ we can write
\begin{align}
 \Delta(C_{ij})=\braket{\psi_i}{\psi_j}  \prod_{l=2}^{k-2} \braket{\psi_{\alpha_l}}{\psi_{\alpha_{l+1}}}.
\end{align}
Therefore, we can fix the phase of every nonzero inner product that is not in $\mathcal{T}^\Psi$ as
\begin{equation}
    \braket{\psi_i}{\psi_j}=\frac{\Delta(C_{ij})}{  \prod_{l=2}^{k-2} \braket{\psi_{\alpha_l}}{\psi_{\alpha_{l+1}}}}.
\end{equation}

Thus, all matrix elements of the so-constructed Gram matrix $\tilde{G}^{\Psi}_{ij}=\braket{\psi_i}{\psi_j}$ are expressed via Bargmann invariants of degree at most $N$. The total number of invariants required to construct $\tilde{G}^{\Psi}$ in this way is upper-bounded by $(N-1)^2$, as we first measure $N(N-1)/2$ degree-2 invariants to obtain the frame graph, and then at most $(N-1)(N/2-1)$ additional invariants to fix the phases of the non-null elements of $\tilde{G}^{\Psi}$ corresponding to edges $\{i,j \} \notin \Gamma^\Psi$.
\end{proof}

\begin{figure}[t]
    \centering
    \includegraphics[width=0.3\textwidth]{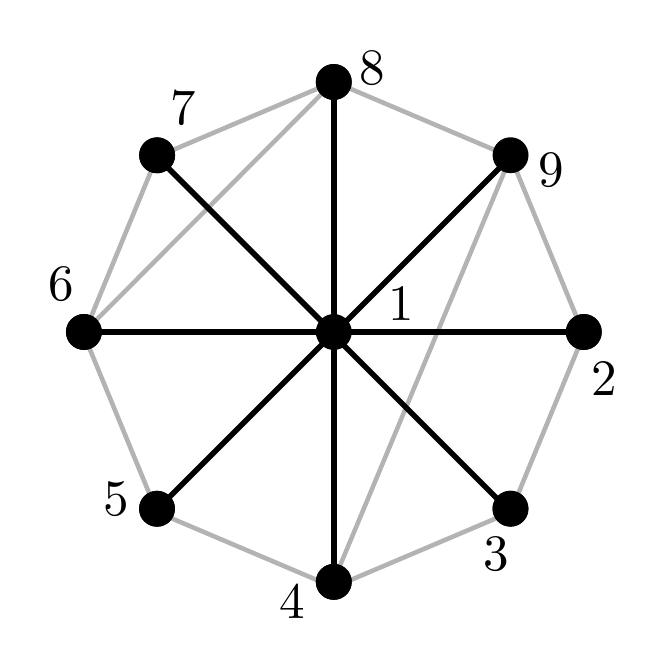}
    \caption{Frame graph for the example described in the text. Edges in black correspond to the choice of spanning tree, in this case a star graph.}
    \label{fig:concept}
\end{figure}


We remark that in general we cannot characterize the relational properties of even single-qubit states without using Bargmann invariants of degree $m>2$. As an example, take two triples of single-qubit states $\Psi=(\ketbra{X_+}{X_+},\ketbra{Y_+}{Y_+}, \ketbra{Z_+}{Z_+}) $, $\Phi=(\ketbra{X_+}{X_+},\ketbra{Y_-}{Y_-}, \ketbra{Z_+}{Z_+}) $, where $\ketbra{P_\pm}{P_{\pm}}$ are projectors on Pauli eigenvectors with $\pm 1$ eigenvalues. All overlaps within each tuple have the same value $\Delta_{ij}= 1/2$. Therefore, the tuples only differ in that $\Delta_{123}^\Psi = \frac{1+i}{4} \neq   \frac{1-i}{4} =\Delta_{123}^\Phi$.

Let us now consider an application of  \cref{thm:PureSates} to a $N=9$ state tuple $(\psi_i)_{i=1\ldots 9}$. The first step is to construct the frame graph. This is done by measuring all $N(N-1)/2 = 36$ overlaps $\left|\left\langle \psi_i |\psi_j \right \rangle \right|$. Suppose the frame graph is the one represented in \cref{fig:concept} (recall that missing edges correspond to null overlaps). One choice of spanning tree is a star graph, where vertex 1 is connected to every other. We can fix the phases so that every degree-2 invariant $\left\langle \psi_1 |\psi_i \right \rangle > 0$. Finally, we consider all cycles that can be built by adding edges from the original frame graph to the spanning tree (in \cref{fig:concept}, these correspond to gray edges). In each such cycle, there is only one new edge whose phase has yet to be fixed; in this example only 3-cycles appear. We then measure the degree-3 Bargmann invariants $\Delta_{1 i j}$ for all distinct pairs $\{i,j\}$, and attribute its phase to edge $\{i,j\}$. There are 28 of these for this example size, though in  \cref{fig:concept} only 10 of them are nonzero. This procedure leads to the Gram matrix:
\begin{equation}
    \tilde{G}^\Psi_{ij}=
    \begin{cases}
       \sqrt{\Delta_{ii}}= \left|\left\langle \psi_i |\psi_i \right \rangle \right|=1 & \text{for } i=j \\
      \sqrt{\Delta_{1j}} =\left|\left\langle \psi_1 |\psi_j \right \rangle \right|, &  \text{for } i=1, \\ 
     \frac{\Delta_{ij1}}{\sqrt{|\Delta_{1j}| |\Delta_{1i}|}} & \text{for }j > i. \label{eq:gram}
    \end{cases}
  \end{equation}
The gauge choices involved in defining the Gram matrix given in Eq.\ (\ref{eq:gram}) mean that all parameters in $\tilde{G}^\Psi$ are physically meaningful quantities, measurable using the cycle test circuits and fully characterizing relative information among states in $\Psi$. 

The choice of spanning tree is not unique. In the above example, we could have chosen instead e.g. a simple path starting in vertex 1 and visiting all vertices sequentially up to vertex 9. Building the Gram matrix $\tilde{G}^\Psi$ using this path requires the same number of invariants, but they are more experimentally involved to measure, as some will have degree $> 3$. An interesting open question is to find alternative characterization procedures for general scenarios that minimize the experimental effort required.



\textit{Robustness, and mixed states.}  
Any implementation of cycle tests will have to deal with experimental imperfections, and the fact that the prepared states will, at best, be mixed-state approximations of the ideal pure states. In Appendix \ref{app:stability} we give an argument that, as long as the prepared density matrices have high purity, then the Bargmann invariants measured will be close to the Bargmann invariants of the target pure states (understood to be the projector onto the eigenvector corresponding to the largest eigenvalue of a  density matrix). Furthermore, in Appendix \ref{app:pustab} we show that if Bargmann invariants evaluated for two tuples of pure states $\Psi= (\psi_i)_{i=1}^N$, $\Phi= ( \phi_i)_{i=1}^N$ differ from each other only slightly, then states comprising them can be \emph{approximately transformed onto each other}. Specifically, using the encoding of Bargmann invariants into Gram matrices we show that there exists a unitary transformation $U$ that \emph{approximately}  transforms $\Psi$ onto $\Phi$ in the sense that the average infidelity is upper bounded by
\begin{equation}
   \frac{1}{N} \sum_{i=1}^N \left(1-\Tr( U \psi_i U^\dagger \phi_i  )\right) \leq \frac{C_\Psi}{N}  \|\tilde{G}^\Psi-\tilde{G}^\Phi\|^2_2 \ ,
\end{equation}
where $C_\Psi$ is a constant (depending on $\Psi$), $\tilde{G}^\Psi,\tilde{G}^\Phi$ are Gram matrices from the proof of Theorem \ref{thm:PureSates}, and $\|X\|_2=\sqrt{\Tr(X^2)}$ is the Hilbert-Schmidt norm. This shows that Bargmann invariants offer a robust characterization of PU equivalence of tuples of pure states.

Identifying complete sets of PU-invariants for tuples of mixed states is a much more challenging problem.  If we use mixed states $\rho_i\ (i=1,2,\dots ,m)$ as inputs to the cycle test circuits of Fig. \ref{fig:swap},  the output gives an operational meaning to mixed-state Bargmann invariants: $\Delta_{12...m}=\Tr\left(\rho_1 \rho_2 \ldots \rho_m \right)$. In Appendix \ref{app:mixedInvariants} we show that for a tuple of $N$  mixed states $\mathcal{S}=(\rho_i)_{i=1}^N$  in dimension $d$  Bargmann invariants of degree $m\leq d^2$ form a complete set of invariants characterizing PU invariant properties of $\mathcal{S}$. Moreover, the number of independent invariants can be chosen to be $Nd^2+1$. These results follow from known results from invariant theory \cite{PROCESI1976306,formanekpolynomial,wigderson2019mathematics}, which we (partially) reprove with the help of Schur-Weyl duality.

\textit{Application: Testing for linear independence.}  Whether a set of states is linearly independent (LI) is clearly a unitary-invariant property. LI states can be probabilistically cloned, and can be unambiguously discriminated. Recognizing linear dependence of a set of states enables us to find a more compact description of the set; in machine learning, this corresponds to a procedure for dimensionality reduction. Using a complete set of Bargmann invariants, we can compute the hypervolume generated by the set of states to determine whether they are linearly independent. For the case of $m$ non-orthogonal vectors, we can use the Gram matrix encoding of the complete set of invariants we have described to write the linear independence condition simply as: $\det(G)>0$. In Appendix \ref{app:li} we calculate $\det(G)$ explicitly as a function of the invariants for the cases of 3 and 4 non-orthogonal states. 

\textit{Application: Basis-independent imaginarity witnesses.} Measuring Bargmann invariants of degree $3$ and above can serve as a basis-independent witness that the input quantum states have complex-valued amplitudes. The simple reason is that $\Delta_{12\dots m}$ is a gauge-invariant quantity, so if we measure $\Im(\Delta_{12\dots m}) \neq 0$, some of the density matrices at the input must necessarily have complex-valued elements, independently of the basis chosen. As an example, it is easy to check that the 3 Pauli operator eigenstates $\ket{A}=\ket{0}, \ket{B}=\frac{1}{\sqrt{2}}(\ket{0}+\ket{1}), \ket{C}=\frac{1}{\sqrt{2}}(\ket{0}+i\ket{1})$ achieves the maximum value of $\Im(\Delta_{ABC})=1/4$ for this 3-state scenario. In Appendix \ref{app:imaginarity} we describe a tuple of $n$ single-qubit states that are vertices of a regular spherical $n$-gon on the Bloch sphere, for which $\Im(\Delta)\to 1$ as $n \to \infty$. 

Experimental tests have been recently proposed to discriminate whether complex  amplitudes (or ``imaginarity'') are an essential feature of quantum theory \cite{Renou21, Wu21}. The application of our proposal as an imaginarity witness differs fundamentally from the proposals of \cite{Renou21,Wu21}. In particular, though basis-independent, our approach is not device-independent. Concretely, we can only interpret the outcome of the circuit of \cref{fig:swap}(b) as measuring a set of Bargmann invariants if the input is promised to be a product state. Therefore, we cannot rule out an adversarial preparation of the states that might fool our imaginarity witness by using an entangled real-amplitude state, similar to techniques of \cite{McKague2009}.

\textit{Application: Basis-independent coherence witnesses.} If the $m$ states $\rho_i$ are simultaneously diagonalizable in a single reference basis, the degree-$m$ Bargmann invariant $\Delta_{12\dots m}=\Tr(\prod_i \rho_i)$ is simply the probability that independent measurements of the reference basis on all states give the same outcome \cite{GalvaoB20}. Hence, for diagonal states these invariants must be real and in the range $\Delta_{12\cdots m} \in [0,1]$. Measuring $\Im(\Delta_{12\dots m})\neq 0$, or real values outside of the allowed range for diagonal states, serves as a basis-independent witness of coherence, generalizing the overlap-based witnesses of \cite{GalvaoB20}. Note that witnesses of imaginarity, as discussed in the previous section, are also witnesses of coherence, as complex-valued entries of a density matrix must be off-diagonal. In Appendix \ref{app:cohe} we look at the simplest scenarios where coherence can be witnessed by Bargmann invariant measurements when we have sets with either 2 or 3 states.

\textit{Application: Characterization of photonic indistinguishability.} The formalism of Bargmann invariants and projective unitary invariance is particularly well-tailored to experiments with single photons in linear optics \cite{SimonM93}. Each photon in an $m$-photon experiment is described by a state over its internal degrees of freedom - such as time-of-arrival, polarization, and orbital angular momentum \cite{Tichy2015}. Suppose we apply some linear-optical transformation to these photons that is insensitive to their internal degrees of freedom, and measure the occupation number at the output of the device. The outcome statistics will not depend on exactly which internal state a particular photons is in, as the apparatus is blind to this information, but it will depend on the relational information, such as imperfect overlaps, which can cause the photons to be effectively partially distinguishable. In other words, the output probabilities of such an experiment can only depend on the Bargmann invariants of the set of states (as well as the linear-optical transformation). This can be seen in several recent experimental and theoretical works concerning genuine 3- and 4-photon indistinguishability \cite{Giordani21,GalvaoB20,Menssen2017,Jones2020,Minke21}.  

In \cite{GalvaoB20}, it was shown how measurements of some overlaps among a set of $m$ states can result in nontrivial lower bounds for all overlaps. Bargmann invariants of higher degree can be useful in the same way,   as by definition $|\Delta_{12\dots m}|^2$ is a lower bound to overlaps of all neighboring states in the associated cycle. This suggests improved designs for photonic measurements of Bargmann invariants may help in the characterization of multi-photon sources.


\textit{Conclusion.}
We reviewed how Bargmann invariants encode the projective-unitary invariant properties of a set of states. For the case of pure states, we have shown how to represent the complete information in a Gram matrix of inner products, written in terms of invariant quantities only, representing relational information about the states in the set. We have proposed a way of  measuring real and imaginary parts of arbitrary Bargmann invariants using cycle tests, a natural generalization of the SWAP test. We discussed several applications: coherence and imaginarity witnesses, testing for linear independence, and characterization of multi-photon indistinguishability. Open problems include finding efficient NISQ methods for measuring and using Bargmann invariants in algorithmic applications, studying possible applications to self-testing in a semi-device independent  context \cite{selftesting2018,Miklin2021},  and developing further the resource theories of imaginarity and coherence based on Bargmann invariants, whose first steps we have described.

\begin{acknowledgments}
 We thank Adam Sawicki for his comments regarding the role of Schur-Weyl duality for the problem of PU equivalence of mixed states. MO  acknowledges support by the Foundation for Polish Science through the TEAM-NET project (contract no. POIR.04.04.00-00-17C1/18-00). 
DJB acknowledges support from Instituto Nacional de Ci\^{e}ncia e Tecnologia de Informação Quântica (INCT-IQ, CNPq) and FAPERJ. EFG acknowledges funding of the Portuguese institution FCT – Fundação para Ciência e Tecnologia via project CEECINST/00062/2018. This work was supported by the H2020-FETOPEN Grant PHOQUSING (GA no.: 899544).
\end{acknowledgments}

\appendix

\section{Efficient constructions for cycle test circuits} \label{app:cycle}

In the main text we described how to use an $m$-state cycle test (cf.\ Fig.\ 1 of main text) to measure the $m$th order Bargmann invariant of the corresponding set of states. The main ingredient of the cycle test is a controlled permutation gate which applies a cyclic permutation when the control qubit is in the $\ket{1}$ state (and does nothing otherwise). We now describe two efficient ways of implementing this controlled cycle permutation. For simplicity, we consider first how to decompose a cyclic permutation in terms of pairwise permutations (i.e.\ transpositions), without accounting for the control qubit, to which we return at the end of this Section.

\begin{figure}[t]
    \centering
    \includegraphics[width=0.4\textwidth]{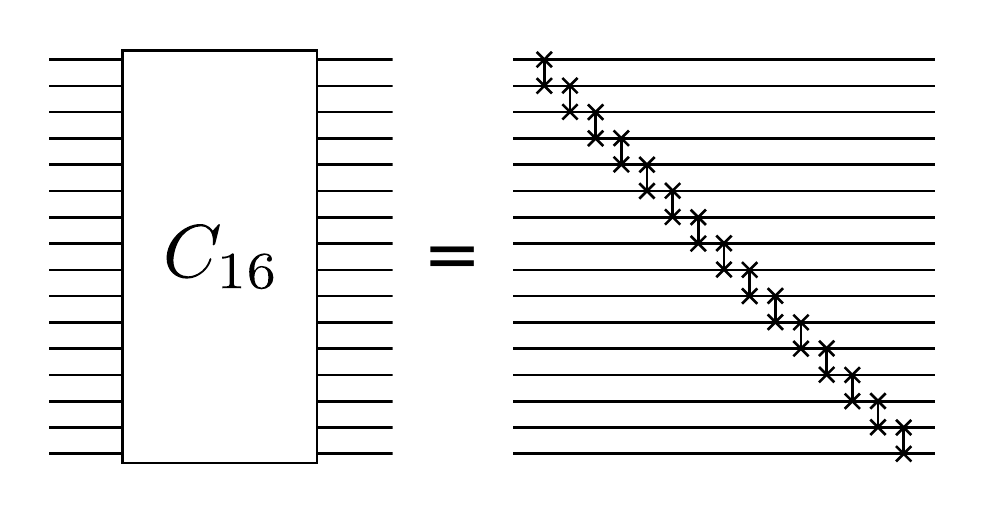}
    \caption{Decomposition of a 16-qubit cyclic permutation in terms of nearest-neighbor SWAP gates.
    }
    \label{fig:nntransp}
\end{figure}

Let us denote $C_k = (1, 2, 3, \ldots, k)$ as the cyclic permutation on $k$ elements, and $(i, k)$ as a permutation of two elements (or transposition) $i$ and $k$. A well-known decomposition of a cyclic permutation in terms of transpositions is given by
\begin{equation}
    (1, 2, 3, \ldots, k) = (1, 2)(2, 3)\ldots (k-1, k)
\end{equation}
This decomposition uses only adjacent transpositions, which is equivalent to swapping only nearest-neighbor qubits at a time. It leads to a circuit as exemplified in Fig.\ \ref{fig:nntransp}. This circuit uses the optimal number of $(k-1)$ nearest-neighbor transpositions and has depth $(k-1)$. That $(k-1)$ is the optimal number follows from the need to move the first element to the last position, which requires exactly $(k-1)$ nearest-neighbor transpositions.

It is possible to give an alternative decomposition that has much smaller depth, at the cost of using longer-range transpositions. To see it, consider first the following Lemma:
\begin{lemma} \label{lem:cycle}
Consider the following two disjoint cycles $(1,2,\ldots k)$ and $(k+1, k+2, \ldots, n)$. We have the following decomposition:
\begin{equation} \label{eq:cycledec}
    C_n (1,\ldots,n) = (1,k+1)(1,2,\ldots k)(k+1, k+2, \ldots, n)
\end{equation}
\end{lemma}
The proof of this Lemma is straightforward, it simply requires checking the images of numbers $1, k, k+1$, and $n$ under the sequence of permutations, as all other numbers get trivially mapped to the correct places. 

\begin{figure}[t]
    \centering
    \includegraphics[width=0.4\textwidth]{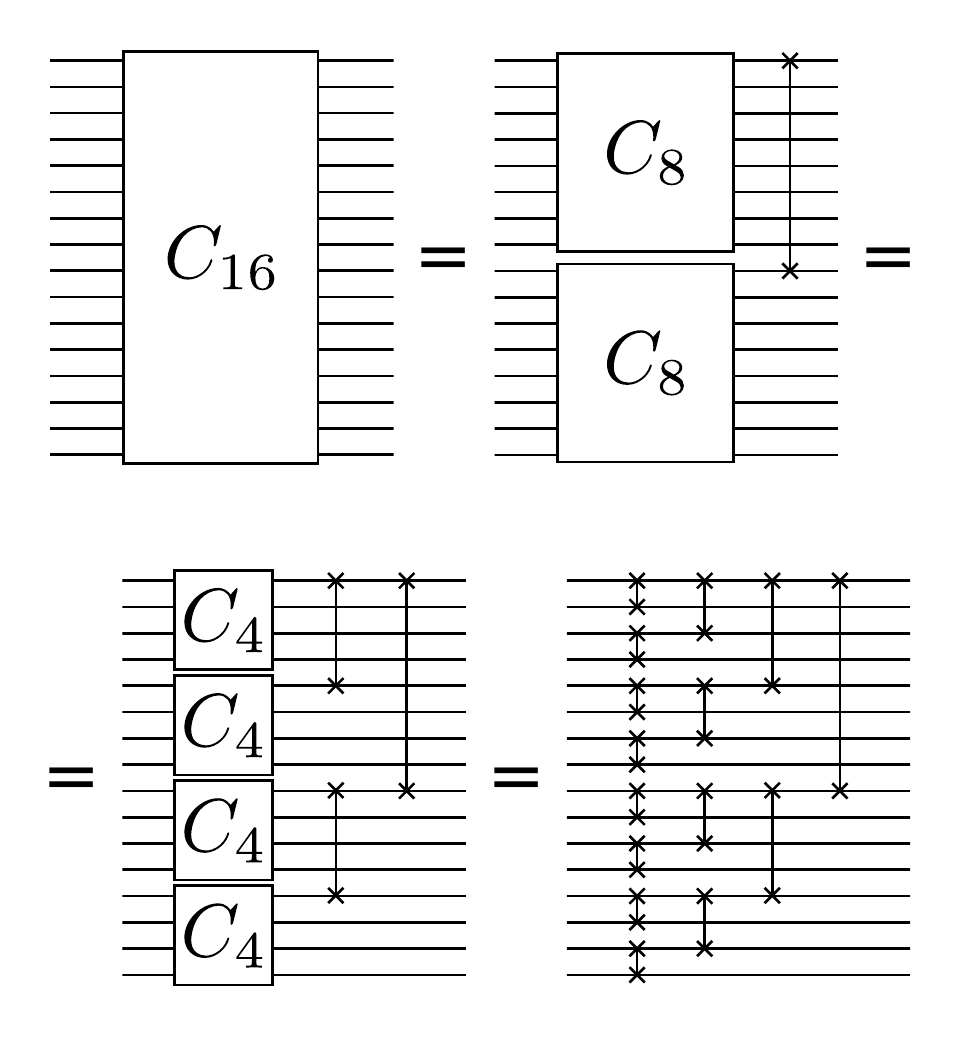}
    \caption{Decomposition of a 16-qubit cyclic permutation in terms of  SWAP gates. This circuit is much more parallelized compared to that of \cref{fig:nntransp}, although the trade-off is the requirement of long-range swaps.}
    \label{fig:cycles2}
\end{figure}

A very efficient decomposition of $C_m$ can now be obtained by applying \cref{lem:cycle} repeatedly. Suppose, for simplicity, that $m$ is a power of 2. Then we can first decompose $C_m$ into two cycles acting on $m/2$ elements followed by one transposition, i.e.\ $C_m = (1, m/2+1)(1, 2, \ldots, m/2)(m/2+1,m/2+2\ldots m)$. We then break each cycle of size $m/2$ again in two cycles, acting each on $m/4$ elements. This divide-and-conquer algorithm is iterated until all that is left are transpositions. It is easy to see that this procedure terminates after O$(\log m)$ steps, and hence that the final circuit has depth O$(\log m)$ as well. An example of this for $m=16$ is shown in \cref{fig:cycles2}. 

These are not the only decompositions of the cycle permutation in terms of transpositions, but they represent two extreme regimes. In the first, we have only nearest-neighbor SWAP gates, but the depth of the circuit is large. In the second, the depth is logarithmic, but longer-range gates are required. Which is the most efficient depends on the constraints of the particular architecture being used---the latter has the benefit of mitigating accumulation of errors in the circuit, but in particular architectures long-range gates might not be an option.

Let us now return to the issue of decomposing a \emph{controlled} cycle permutation. At first glance, neither decomposition we described has the promised properties. The circuit in \cref{fig:nntransp} is no longer composed only of nearest-neighboring controlled-SWAP (or Fredkin) gates, since the first qubit must control the swapping of states that are distant to it. At the same time, the circuit of \cref{fig:cycles2} does not have logarithmic depth since every gate in the circuit must be controlled by the same control qubit, which presumably means that all gates must be done sequentially. Nonetheless, both circuits can be adapted to restore the desired properties, as follows.

\begin{figure}[t]
    \centering
    \includegraphics[width=0.35\textwidth]{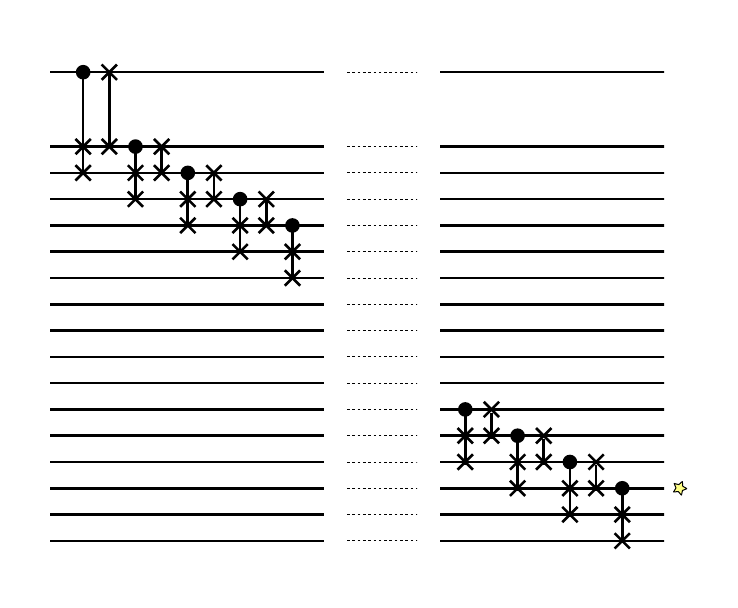}
    \caption{An adaptation of the circuit of \cref{fig:nntransp} that obtains a controlled cycle permutation with SWAP and Fredkin gates acting only on neighboring state. The first qubit is swapped along the 1D chain as each Fredkin gate is applied. The star at the output indicates where the control qubit is effectively left at the end of the circuit.}
    \label{fig:cyclesnnfix}
\end{figure}

In \cref{fig:cyclesnnfix} we show how the controlled version of circuit in \cref{fig:nntransp} can be implemented using only gates (SWAP and Fredkin) that act on neighboring qubits.  

To obtain a controlled version of \cref{fig:cycles2} that has depth O$(\log m)$, as in the original, we need to include more ancilla qubits to parallelize the control operation. Instead of controlling the permutations with a single $\tfrac{1}{\sqrt{2}} (\ket{0}+\ket{1})$, we need to build an $m/2$-qubit GHZ state $\tfrac{1}{\sqrt{2}} (\ket{0}^{\otimes m/2}+\ket{1}^{\otimes m/2})$, such that each qubit in this state can control a different SWAP gate in parallel. The whole construction still has logarithmic depth due to the follow Lemma (also cf.\ Figure 2 in \cite{Cruz2019}):
\begin{lemma} \cite{Cruz2019} \label{lem:GHZ}
The $m$-qubit GHZ state:
\begin{equation}
    \frac{1}{\sqrt{2}} (\ket{0}^{\otimes m}+\ket{1}^{\otimes m})
\end{equation}
can be generated by a circuit that acts on the $\ket{0}^{\otimes m}$ state, is composed only of a single Hadamard and a sequence of CNOT gates, and has depth O$(\log m)$.
\end{lemma}
\begin{proof}
We initialize $m$ qubits in the $\ket{0}^{\otimes m}$ state, and apply a Hadamard gate to the first one. Our goal is now to apply a sequence of CNOT gates that will flip every qubit conditioned on the first qubit being $\ket{1}$ and does nothing otherwise. This can be done as follows. In the first round, we apply a CNOT gate from qubit 1 to qubit 2. In the second round, we apply two CNOTs, from qubits 1 and 2 to qubits 3 and 4, respectively. We repeat this procedure by reusing previously flipped qubits as controls for the next layer. The number of qubits flipped in this manner doubles with each layer. Hence, it follows that this will generate an $m$-qubit GHZ state in O$(\log m)$ layers of two-qubit gates.
\end{proof}
 
For the controlled version of \cref{fig:cycles2} to work properly when using the GHZ state as the ancilla, we need to perform the operation described in the proof of \cref{lem:cycle} in reverse at the end of the cycle test protocol. This is required to disentangle the ancillas with the control qubit, such that interference can occur properly at the last $H$ gate in \cref{fig:swap}. In any case, it follows that the entire controlled-cycle protocol works in depth O$(\log m)$.

\section{Bargmann invariants for almost pure states }\label{app:stability}

In experiments to measure Bargmann invariants, the preparations will result in mixed states that only approximate the target pure states.
 Here we prove two lemmas that show that Bargman invariants computed for mixed states with high purity will have values close to those of the ideal, target pure states. 

\begin{lemma} Let $\rho$ be a mixed state on  $\mathbb{C}^d$ (with arbitrary dimensionality $d$), with $\Tr(\rho^2)\ge 1-\epsilon$ for some $\epsilon \le 1/2$. Let $\psi$ be the normalized projector onto the eigenspace with largest eigenvalue of $\rho$. Then $||\rho-\psi||_1 \le 2\epsilon$, where $\|X\|_1=\Tr(\sqrt{XX^\dagger})$ denotes the trace norm.
\end{lemma}
\begin{proof}: Let $\lambda$ be the maximal eigenvalue of $\rho$. Write $\rho=\lambda \psi + (1-\lambda)G$ (with $G \perp \psi$). Then
\begin{align}
    \Tr(\rho^2)&=\lambda^2 + (1-\lambda)^2\Tr(G^2)\\&\le \lambda^2+(1-\lambda)^2\\&=2\lambda^2-2\lambda+1.
\end{align}
(as the maximal value for $\Tr(G^2)=1$). Then we have:
\begin{align}
    &\Tr(\rho^2)\ge 1-\epsilon  \notag\\
    &\Rightarrow 2\lambda(\lambda-1) \ge -\epsilon.\\
    &\Rightarrow \lambda  \ge \frac{1}{2}\left(1+\sqrt{1-2\epsilon} \right).
\end{align}
As a result we obtain
\begin{align}
    \|\rho-\psi\|_1&=\|(\lambda - 1)\psi+(1-\lambda)G \|_1 \\ 
    &=2(1-\lambda) \\ &\le 1-\sqrt{1-2\epsilon} \\ &=\frac{2\epsilon}{1+\sqrt{1-2\epsilon}}\le 2\epsilon,
\end{align}
as we wanted to demonstrate. 
\end{proof}

\begin{lemma}\label{lem:stababBERG} Assume we have two $m$-tuples of (generally mixed) quantum states $( \rho_1,\ldots,\rho_m)$,  $( \sigma_1,\ldots, \sigma_m)$  which satisfy  $||\rho_i - \sigma_i ||_1 \le \delta$ for $i=1,2,\ldots,m$. Then,
\begin{equation}
    \left|\Tr\left(\prod_{i=1}^m \rho_i\right)-\Tr\left(\prod_{i=1}^m \sigma_i\right)\right|\leq m \delta\ .
\end{equation}
\end{lemma}
\begin{proof}
We use the formula expressing Bargmann invariant in terms of expectation value of cycle permutation $C_m$ on tensor product of the relevant quantum states
\begin{equation}
\Tr\left(\prod_{i=1}^m \sigma_i\right)=\Tr\left(C_m \otimes_{i=1}^m \sigma_i\right)\ .
\end{equation}
It follows that 
\begin{align}
     &\left|\Tr\left(\prod_{i=1}^m \rho_i\right)-\Tr\left(\prod_{i=1}^m \sigma_i\right)\right| \\
     &= 
     \left| \Tr\left[C_m \left(\bigotimes_{i=1}^m \rho_i - \bigotimes_{i=1}^m \sigma_i\right)\right] \right| \\
   & \leq \|C_m\|_\infty  \|\otimes_{i=1}^m \rho_i - \otimes_{i=1}^m \sigma_i \|_1 \label{eq:traceNORM}\\
   & \leq m \delta\ \label{eq:subadd},
\end{align}
where in \eqref{eq:traceNORM} we employed  a well-known inequality $|\Tr(XY)|\leq \|X\|_\infty \|Y\|_1$, while in \eqref{eq:subadd} we used $\|C_m\|_\infty=1$ ($C_m$ is a unitary operator) and subaddtitivity of 1-norm: 
\begin{equation}
    \|\otimes_{i=1}^m \rho_i - \otimes_{i=1}^m \sigma_i \|_1 \leq \sum_{i=1}^m \|\rho_i -\sigma_i\|_1 \leq m \delta\ .
\end{equation}
\end{proof}

Combining the above two lemmas we get that Bargmann invariants computed for mixed states with high purity approximate Bargmann invariants computed for collection of pure states corresponding to projectors onto eigenvectors corresponding to the highest eigenvalue. 

\begin{corr}
Let $\mathcal{S}=\lbrace\rho_1,\ldots,\rho_N\rbrace$ be collection of mixed states that satisfy $\Tr(\rho_i^2)\geq 1-\epsilon$ for $\epsilon\in[0,1/2]$ and let $\Psi=\lbrace \psi_1,\ldots,\psi_N\rbrace$, where $\psi_i$ is a projector onto  an eigenvector corresponding to the largest eigenvalue of  $\rho_i$. Then for for arbitrary sequence $i_k\in\lbrace{1,\ldots,N\rbrace}$ of length $m$ we have
\begin{equation}
|\Tr(\rho_{i_1}\rho_{i_2} \ldots \rho_{i_m}) - \Tr(\psi_{i_1}\psi_{i_2}\ldots \psi_{i_m}) | \leq 2 m \epsilon.
\end{equation}
 
\end{corr}

\section{Stability of PU equivalence from Bargmann invariants} \label{app:pustab}

The following result is a quantitative variant of Theorem \ref{thm:PureSates}---it shows two collections of pure states $\Psi$,$\Phi$ can be transformed onto each other approximately provided the relevant Bargmann invariants approximately match. 

\begin{lemma}[Robust Characterization of PU equivalence via Bargmann invariants]
Let $\Psi, \Phi$ be two tuples of pure states with $m$ elements each: 
\begin{equation}
    \Psi = ( \psi_1,\psi_2,\ldots, \psi_N ),  \Phi = ( \phi_1,\phi_2,\ldots, \phi_N )\ .
\end{equation}
 Moreover, assume that Bargmann invariants used to construct $\tilde{G}^\Psi,\tilde{G}^\Phi$ are such that
  \begin{equation}
    \mathrm{rank}(\tilde{G}^\Psi)\geq \mathrm{rank}(\tilde{G}^\Phi),\   \|\tilde{G}^\Psi-\tilde{G}^\Phi\|_2 \|(\tilde{G}^\Psi)^\ddag\|_\infty^{\frac{1}{2}}\leq \frac{1}{2},
  \end{equation}
where, $\|\cdot\|_\infty, \|\cdot\|_2$ denote operator and Frobenius norm respectively and $A^\ddag$ is a Moore-Penrose pseudo inverse \cite{ProcrustesBound}. Then, there exist a unitary $U$ such that
\begin{equation}\label{eq:appBound}
   \frac{1}{N} \sum_{i=1}^N \left(1-\Tr( U \psi_i U^\dagger \phi_i  )\right) \leq \frac{C_\Psi \|\tilde{G}^\Psi-\tilde{G}^\Phi\|^2_2}{N} \,
\end{equation}
where $C_\Psi= (1+\sqrt{2})^2 \|(\tilde{G}^\Psi)^\ddag\|_\infty$. 
\end{lemma}

In other words, closeness of Bargmann invariants defining Gram matrices $\tilde{G}^\Psi,\tilde{G}^{\Phi}$ implies that under a suitable unitary transformation $U$ the  average-state infidelity between states $\Psi$ and $\Phi$ can be upper bounded as in \eqref{eq:appBound}. 
\begin{remark}
We note that in order to construct Gram matrices $\tilde{G}^\Psi, \tilde{G}^\Phi$ we need to assume that $\tr(\psi_i \psi_j) \neq 0$ if and only if $\tr(\phi_i \phi_j) \neq 0$ (in other words that frame graphs  defined in the proof of Theorem \ref{thm:PureSates} are identical, $\Gamma^\Psi=\Gamma^\Phi$). Furthermore, we assume that both matrices $\tilde{G}^\Psi, \tilde{G}^\Phi$ are constructed using the same spanning tree of the corresponding frame graph.
\end{remark}

\begin{proof}
The proof follows from the stability of the celebrated Procrustes problem discussed in \cite{ProcrustesBound}. In Theorem 1 from that work it was shown that if Gram matrices of two tuples of vectors  
\begin{align}
\mathcal{A}&=\left(\ket{\psi_1}, \ket{\psi_2},\dots \ket{\psi_N} \right),\\
\mathcal{B}&=\left(\ket{\phi_1}, \ket{\phi_2},\dots \ket{\phi_N} \right),
\end{align}
satisfy 
  \begin{equation}
    \mathrm{rank}(G^\mathcal{A})\geq \mathrm{rank}(G^\mathcal{B}),\   \|G^\mathcal{A}-G^\mathcal{B}\|_2 \|(G^\mathcal{A})^\ddag\|_\infty^{\frac{1}{2}}\leq \frac{1}{2}\ ,
  \end{equation}
then there exist a unitary transformation $U$ such that 
\begin{equation}\label{eq:appBoundstate}
   \frac{1}{N} \sum_{i=1}^N \|U\ket{\psi}-\ket{\phi}\|^2   \leq \frac{(1+\sqrt{2})^2 \|(G^\mathcal{A})^\ddag\|}{N}  \|G^\mathcal{A}-G^\mathcal{B}\|^2_2 \ .
\end{equation}
We can now apply this result directly to Gram matrices $\tilde{G}^\Psi$, $\tilde{G}^\Phi$ constructed form Bargmann invariants (and vector representatives of states form $\Psi,\Phi$ compatible with these Graam matrices). To pass from \eqref{eq:appBoundstate} to \eqref{eq:appBound} we note that for normalized wavefunctions $\ket{\psi},\ket{\phi}$ we have $  1-\Tr( U \psi_i U^\dagger \phi_i  ) \leq  \|U\ket{\psi_i}-\ket{\phi_i}\|^2$.
\end{proof}

\section{Complete PU invariants for mixed states} \label{app:mixedInvariants}

In this Appendix we show that Bargmann invariants can be used to decide unitary equivalence of two tuples of (generically mixed) quantum states $\mathcal{S}=({\rho_i)}_{i=1}^N, \mathcal{S'}=({\sigma_i)}_{i=1}^N$. This result was previously established in \cite{PROCESI1976306}. Here, we present a self-consistent proof based on Schur-Weyl duality. After presenting this general result, we remark on the connection of the unitary equivalence problem with the question of simultaneous conjugation of tuples of matrices. This connection allows us to state bounds on the degrees and number of independent Bargmann invariants. 

We begin by introducing polynomial invariants of linear group actions of compact groups in real vector spaces (see e.g. \cite{fulton1991representation,goodman2009symmetry} for an introduction to representation theory and invariant theory). Let $G$ be a compact group and let 
\begin{equation}
    \Pi:G\ni g\mapsto \Pi_g\in \mathrm{End}(V)
\end{equation}
be a representation of $G$ in a finite dimensional real vector space $V$ \footnote{Typically, in the context of invariant theory one  considers \emph{complex} vector spaces. In that setting the derivation of invariant polynomials is slightly more complicated - this is a consequence of the fact that even for one dimensional complex space $\mathbb{C}$, general polynomials can be understood as two variate polynomials in $\mathrm{Re}z,\mathrm{Im}z$. We decided to limit ourselves to \emph{real} vector spaces to avoid unnecessary technical complication. }. Since $G$ is compact we can assume that $\Pi_g$ is an orthogonal transformation. Homogeneous polynomials of degree $m$ in $V$ are mappings of the form 
\begin{equation}
    p(v)=\langle \tilde{p}, v^{\otimes m}\rangle\ ,
\end{equation}
where $\langle\cdot,\cdot\rangle$ is the invariant inner product in $V^{\otimes k}$ (induced by the inner product on $V$), and $\tilde{p}\in V^{\otimes m}$ is a tensor encoding the polynomial $p$. We will denote by $\mathbb{P}_k(V)$ the space of real homogeneous polynomials on $V$ of degree $m$. By $\mathbb{P}(V)$ we will denote the space off all polynomials on $V$.  

Invariant homogeneous polynomials of degree $m$ are polynomials that additionally satisfy 
\begin{equation}\label{eq:invPOLY}
 p(\Pi_g v)= p(v) \ \text{for every } v\in V, g\in G\ .
\end{equation}
Using orthogonality of $\Pi_g$ it can be shown that \eqref{eq:invPOLY} is equivalent to requiring
\begin{equation}\label{eq:invPOLYalt}
\Pi_g^{\otimes m} \tilde{p}= \tilde{p} \ \text{for every } g\in G\ .
\end{equation}
We will denote the set of all invariant polynomials by $\mathbb{P}(V)^G$.
It is a well known result in invariant theory that in the case of compact groups we can use invariant polynomials $\mathbb{P}(V)^G$ to decide about equivalence of elements of $V$ under the action of $G$. In other words  for $v,w\in V$ we have  $w=\Pi_g v$, for some $g\in G$ if and only if for every invariant polynomial $p\in\mathbb{P}(V)^G$ we have $p(w)=p(v)$ (see e.g. Theorem 2.4.1 in \cite{vranaPHD} for the elementary exposition in the context of entanglement theory).

The problem of unitary equivalence of $N$ tuples of density matrices  in $\mathbb{C}^d$ can be presented using the formalism outlined above. Let 
\begin{equation}
    V=\overbrace{\mathrm{Herm}(\mathbb{C}^d)\oplus \ldots \oplus \mathrm{Herm}(\mathbb{C}^d)}^{N\  \textrm{times }}\ ,
\end{equation}
where $\mathrm{Herm}(\mathbb{C}^d)$ denotes a (real) vector space of Hermitian matrices on $\mathbb{C}^d$. Every $\mathbf{X}\in V$ can be identified with a tuple of linear operators:
\begin{equation}
    \mathbf{X}=(X_1,\ldots,X_N)\ .
\end{equation}
Let $\mathrm{U}(d)$ denote the group of unitary matrices in $\mathbb{C}^d$. This group can be represented on $V$ via simultaneous conjugation:
\begin{equation}\label{eq:defREP}
\Pi_U\left[(X_1,\ldots, X_N)  \right] = \left(\mathrm{Ad}_U[X_1],\ldots, \mathrm{Ad}_U[X_N] \right)\ , 
\end{equation}
where $\mathrm{Ad}_U[X]=U XU^\dagger$ is the adjoint representation of $\mathrm{U}(d)$ on $\mathrm{Herm}(\mathbb{C}^d)$. We remark that the problem of orbit classification with respect to this so-defined representation of $\mathrm{U}(d)$  corresponds \emph{exactly} to the problem of unitary equivalence of tuples of mixed states $\mathcal{S}=(\rho_1,\ldots,\rho_N),\mathcal{S}'=(\sigma_1,\ldots,\sigma_N)$ discussed in the main text.  Using the unitary equivalence of vector spaces
\begin{equation}
    \overbrace{\mathrm{Herm}(\mathbb{C}^d)\oplus \ldots \oplus \mathrm{Herm}(\mathbb{C}^d)}^{N\  \textrm{times }} \simeq \mathrm{Herm}(\mathbb{C}^d) \otimes \mathbb{R}^N\ ,
\end{equation}
we can represent $\mathbf{X}\in V$ via
\begin{equation}
    \mathbf{X}=\sum_{i=1}^N X_i\otimes \ket{i}\ ,
\end{equation}
where $\lbrace{\ket{i}\rbrace}_{i=1}^N$ is an orthogonal basis of $\mathbb{R}^N$. In this language the representation $\Pi_U$ defined in \eqref{eq:defREP} is equivalent to
\begin{equation}\label{eq:multiplicityAd}
    \Pi_U \simeq \mathrm{Ad}_U\otimes \mathbb{I}_N\ ,
\end{equation}
where $\mathbb{I}_N$ is the identity operator on $\mathbb{R}^N$. We note that an invariant inner product with respect to representation $\Pi_U$ can be chosen as
\begin{equation}\label{eq:innerprod}
 \langle \mathbf{X}, \mathbf{Y} \rangle = \sum_{i=1}^N \Tr(X_i Y_i)\ .
\end{equation}

The equivalence \eqref{eq:multiplicityAd} allows us to conveniently describe conditions that $\tilde{p}\in V^{\otimes m}$ has to satisfy in order for it to define an invariant polynomial of degree $m$ (cf. Eq. \eqref{eq:invPOLYalt}). Indeed, upon reordering the factors in the tensor product we have 
\begin{equation}\label{eq:trensorEquiv}
V^{\otimes m}\simeq\left(\mathrm{Herm}(\mathbb{C}^d)\otimes \mathbb{R}^N \right)^{\otimes m} \simeq  \mathrm{Herm}(\mathbb{C}^d)^{\otimes m} \otimes (\mathbb{R}^N)^{\otimes m}\ ,
\end{equation}
\begin{equation}
    \Pi_U ^{\otimes m} \simeq \left(  \mathrm{Ad}_U\otimes \mathbb{I}_N\right)^{\otimes m} \simeq \mathrm{Ad}_U^{\otimes N} \otimes \mathbb{I}^{\otimes m}_N\ .
\end{equation}
Using \eqref{eq:trensorEquiv} we get that an arbitrary element of $\tilde{p}\in V^{\otimes m}$ can be written as
\begin{equation}\label{eq:explicitFORM}
 \tilde{p} = \sum_{i_1=1}^N \sum_{i_2=1}^N \ldots  \sum_{i_m=1}^N P_{i_1 i_2 \ldots i_m} \otimes \ket{i_1}\ket{i_2}\ldots\ket{i_m}\ ,
\end{equation}
where  $P_{i_1 i_2 \ldots i_m}\in \mathrm{Herm}((\mathbb{C}^d)^{\otimes m})$ is a Hermitian operator acting on $(\mathbb{C}^d)^{\otimes m}$. Note that we implicitly used the equivalence $\mathrm{Herm}((\mathbb{C}^d)^{\otimes m})  \simeq  (\mathrm{Herm}(\mathbb{C}^d))^{\otimes m}$]. Recall that $\tilde{p}$ defines an invariant polynomial $p\in\mathbb{P}_m(V)^{\mathrm{U}(d)}$ if and only if $\Pi_U^{\otimes m} \tilde{p} =\tilde{p}$ for all $Y\in\mathrm{U}(d)$. By virtue of the above this is equivalent to demanding that for all $U$ $\mathrm{Ad}_U^{\otimes m} \otimes \mathbb{I}^{\otimes m}_N \tilde{p}=\tilde{p}$. Using \eqref{eq:explicitFORM} and the fact that vectors $\ket{i_1}\ket{i_2}\ldots\ket{i_m}$ are linearly independent we get that this amounts to requiring
\begin{equation}\label{eq:intermediateCommutation}
 [U^{\otimes m}, P_{i_1 i_2 \ldots i_m}]=0\ ,
\end{equation}
for all $U\in\mathrm{U}(d)$ and all sequences  $i_1,i_2,\ldots,i_m$. From the celebrated Schur-Weyl duality \cite{goodman2009symmetry} it follow that \eqref{eq:intermediateCommutation} is equivalent to $P_{i_1 i_2 \ldots i_m}$ being a linear combination of  permutation operators $\tau$ that permute factors of the space $(\mathbb{C}^d)^{\otimes m}$:
\begin{equation}
    \tau \ket{\psi_1}\ket{\psi_2}\ldots \ket{\psi_m} \eqdef \ket{\psi_{\tau^{-1}(1)}}\ket{\psi_{\tau^{-1}(2)}}\ldots\ket{\psi_{\tau^{-1}(m)}}\ .
\end{equation}
We note that the expansion of $P_{i_1 i_2 \ldots i_m}$ onto permutations $\tau$ must be such that the resulting operator is Hermitian (generically it will be a real linear combination of operators of the form $\tau+\tau^\dagger$, $i(\tau -\tau^\dagger)$, which are themselves Hermitian). 

Let us evaluate a (possibly complex valued) polynomial corresponding to $\tilde{p}_\tau=\tau \otimes \ket{i_1}\ket{i_2}\ldots\ket{i_m}$ on a tuple of operators $\textbf{X}=(X_1,\ldots,X_N)$ which we represent by  $\textbf{X}=\sum_{j=1}^N X_j \otimes\ket{j}$. We get
\begin{equation}\label{eq:evaluation}
 p_\tau(\mathbf{X})= \langle \tilde{p}_\tau , \mathbf{X}^{\otimes m} \rangle =  \left\langle \tau \otimes \ket{i_1} \ket{i_2} \ldots \ket{i_m}  , \mathbf{X}^{\otimes m} \right\rangle\ .
\end{equation}
By expanding $\mathbf{X}^{\otimes m}$ into individual tensors we obtain 
\begin{equation}
    \mathbf{X}^{\otimes m} = \sum_{j_1=1}^N  \sum_{j_2=1}^N \ldots \sum_{j_m =1}^N \left(\bigotimes_{k=1}^m X_{j_k}\right) \otimes \ket{j_1} \ket{j_2}\ldots \ket{j_m}\ .
\end{equation}
Inserting this into \eqref{eq:evaluation} and using the inner product structure \eqref{eq:innerprod} (promoted to $V^{\otimes m}$)  and  orthonormality of vectors $\ket{j_1} \ket{j_2}\ldots \ket{j_p}$ gives
\begin{equation}
    p_\tau(\mathbf{X})= \Tr\left(\tau \bigotimes_{k=1}^m X_{i_k}  \right)\ ,
\end{equation}
where $i_1,i_2,\ldots,i_m$ is an \emph{arbitrary} sequence of numbers form $\lbrace 1, \ldots,N\rbrace$. Using the decomposition of  $\tau$  into disjoint cycles $\tau= C_{k_1} C_{k_2} \ldots C_{k_m}$ (indices $k_i$ denote lengths of disjoint cycles composing $\tau$), we get that 
\begin{equation}
    p_\tau(\mathbf{X})=\Tr\left(\prod_{l_1=1}^{k_1} X_{l_1}\right) \Tr\left(\prod_{l_2=1}^{k_2} X_{l_2}\right) \ldots \Tr\left(\prod_{l_p=1}^{k_p} X_{l_p}\right)\ ,
\end{equation}
for suitably chosen sequences $\lbrace l_1 \rbrace_{l_1}^{k_1}$,...,  $\lbrace l_p \rbrace_{l_p}^{k_p}$, that depend on $\tau$ and the sequence  $i_1,i_2,\ldots,i_m$. This shows that $p_\tau(\mathbf{X})$ can be expressed as a product of Bargmann invariants of degree at most $m$. 

Since an \emph{arbitrary} real polynomial invariant $p\in\mathbb{P}(V)^{\mathrm{U}(d)}$ can be expressed via a suitable linear combination of  $p_\tau$, we finally conclude that Bargmann invariants decide simultaneous unitary equivalence of two tuples of Hermitian operators $\mathbf{X}=(X_1,\ldots,X_N),\mathbf{Y}=(Y_1,\ldots,Y_N)$ operating in $\mathbb{C}^d$. Restricting to density matrices we see that the above considerations proved the desired result for equivalence of tuples of mixed states.
\begin{thm}
Let $\mathcal{S}=(\rho_1,\ldots,\rho_N),\mathcal{S}'=(\sigma_1,\ldots,\sigma_N)$  be two tuples of mixed states on $\mathbb{C}^d$. Then, there exist $U\in\mathrm{U}(d)$ such that for every $i=1,\ldots,N$ 
\begin{equation}
\sigma_i = U\rho_i U^\dagger\ 
\end{equation}
if and only if for every  $m$ and for every sequence $i_1,i_2,\ldots, i_m$ of numbers from $\lbrace 1,\ldots, N\rbrace$ the corresponding Bargman invariants of degree $m$ agree 
\begin{equation}
    \Tr\left(\rho_{i_1} \rho_{i_2} \ldots \rho_{i_m} \right) = \Tr\left(\sigma_{i_1} \sigma_{i_2} \ldots \sigma_{i_m} \right).
\end{equation}
\end{thm}

We conclude this part by remarking about the degree and number of generating polynomials that suffice to decide unitary equivalence of tuples of density matrices in $\mathbb{C}^d$. Bargmann invariants $\Tr(X_{i_1}X_{i_2}\ldots X_{i_m})$ are relevant also for the problem of simultaneous matrix conjugation \cite{PROCESI1976306,formanekpolynomial,wigderson2019mathematics}. This is a related equivalence problem in which the unitary group $\U(d)$ is replaced by the special linear group $\mathrm{SL}(d)$ (a group of $d\times d$ matrices with unit determinant) with a representation on $N$-tuples of matrices defined by 
\begin{equation}
    \Pi_Z\left[(X_1,\ldots, X_N)  \right] = \left( ZX_1 Z^{-1},\ldots, Z X_N Z^{-1} \right)\ ,
\end{equation}
where $X_i \in \mathrm{End}(\mathbb{C}^d)$ and $Z\in \mathrm{SL}(d)$. It turns out that \emph{the same} Bargmann invariants decide the problem of equivalence of (closures of) orbits of group action in this scenario. Because of this the ring of invariant polynomials generated by  $\Tr(X_{i_1}X_{i_2}\ldots X_{i_m})$ has been extensively studied in the context of invariant theory - see Section 13.9.2 in \cite{wigderson2019mathematics} for an accessible exposition. It follows from the discussion contained there that the degrees of Bargmann invariants generating the whole invariant ring   $\mathbb{P}(V)^{\U(d)}$  are at most $m\leq d^2$.   Furthermore, the number of independent invariants can be chosen to be $Nd^2+1$ and can be found efficiently in time polynomial in both  $d$ and $N$.

\section{Conditions for linear independence}
\label{app:li}

As described in the main text, we can parameterize the Gram matrix $G$ so that a complete set of Bargmann invariants are represented in the matrix elements. The condition for linear independence is simply:
\begin{equation}
    \det(G)>0.
\end{equation}
Having measured all the necessary invariants, the determinant of the resulting numerical Gram matrix $G$ can be computed efficiently.

To illustrate the type of dependence on the Bargmann invariants that results, below we will work out the expression for $\det(G)$ for $m=3,4$ normalized pure states (which we assume to be non-orthogonal). We will use the construction of the Gram matrix presented in Theorem 1 of the main text, choosing for the spanning tree the star graph with state $\ket{\psi_1}$ as the central node. For 3 pure states $\ket{\psi_1}, \ket{\psi_2}, \ket{\psi_3}$, we write the Gram matrix:
\begin{align}
G =& 
\begin{pmatrix}
1 & \left|\left\langle \psi_1 |\psi_2 \right \rangle\right| & \left|\left\langle \psi_1 |\psi_3 \right \rangle\right| \notag \\
\left|\left\langle \psi_1 |\psi_2 \right \rangle\right| & 1 & \left\langle \psi_2 |\psi_3 \right \rangle \\
\left|\left\langle \psi_1 |\psi_3 \right \rangle\right| & \left\langle \psi_2 |\psi_3 \right \rangle^*& 1
\end{pmatrix} \\
=&
\begin{pmatrix}
1 & \sqrt{\Delta_{12}} & \sqrt{\Delta_{13}} \\
\sqrt{\Delta_{12}} & 1 & \sqrt{\Delta_{23}}e^{i\phi_{23}} \\
\sqrt{\Delta_{13}} & \sqrt{\Delta_{23}}e^{-i\phi_{23}}& 1
\end{pmatrix},
\end{align}
where we have parameterized $G$ in terms of overlaps, i.e. 2-vertex invariants $\Delta_{ij}=\left|\left\langle \psi_i |\psi_j \right \rangle \right|^2=\frac{|\Delta_{ijj}|^2}{\Delta_{1ii} \Delta_{1jj}}$ and the single 3-state invariant phase $e^{i\phi_{23}}=\Delta_{123}/|\Delta_{123}|$. 
For the case we consider here (no null overlaps), the linear independence condition is then:
\begin{align*}
    \det(G) &>0 \Leftrightarrow \\ 1-(\Delta_{12}+\Delta_{13}+\Delta_{23})&+2\sqrt{\Delta_{12}\Delta_{13}\Delta_{23}}\cos(\phi_{23})>0.
\end{align*}
This condition has appeared previously in the literature, in the context of unambiguous state discrimination \cite{Bergou12}. Note that measurement of overlaps only may suffice to guarantee linear independence, as $1-(\Delta_{12}+\Delta_{13}+\Delta_{23})-2\sqrt{\Delta_{12}\Delta_{13}\Delta_{23}}>0$ implies $\det(G)>0$ independently of the value of $\phi_{23}$. This was the approach to dimension witnesses based on overlaps proposed in \cite{GalvaoB20}, and measured experimentally in \cite{Giordani21}.

The case of 4 states with no null overlaps can be analyzed similarly, using overlaps and phases of degree-3 invariants. The Gram matrix in this case can be written thus:
\begin{equation*}
G = 
\begin{pmatrix}
1 & \sqrt{\Delta_{12}} & \sqrt{\Delta_{13}} & \sqrt{\Delta_{14}}\\
\sqrt{\Delta_{12}} & 1& \sqrt{\Delta_{23}}e^{i\phi_{23}} &\sqrt{\Delta_{24}}e^{i\phi_{24}} \\
\sqrt{\Delta_{13}} &\sqrt{\Delta_{23}}e^{-i\phi_{23}}&1& \sqrt{\Delta_{34}}e^{i\phi_{34}}\\
\sqrt{\Delta_{14}}&\sqrt{\Delta_{24}}e^{-i\phi_{24}}&\sqrt{\Delta_{34}}e^{-i\phi_{34}}&1
\end{pmatrix},
\end{equation*}
where again we have the overlaps $\Delta_{ij}=\left|\left\langle \psi_i |\psi_j \right \rangle \right|^2$ and the 3-state Bargmann invariant phases $e^{i\phi_{ij}}=\Delta_{1ij}/|\Delta_{1ij}|$. 

The linear independence condition can then be stated:
\begin{widetext}
\begin{align*}
    \det(G)>0 \Leftrightarrow 1-(\Delta_{12}+\Delta_{13}+\Delta_{14}+\Delta_{23}+\Delta_{24}+\Delta_{34})+\left(\Delta_{12}\Delta_{34}+\Delta_{13}\Delta_{24}+\Delta_{14}\Delta_{23}\right)+\\
    +2\left[\sqrt{\Delta_{12}\Delta_{13}\Delta_{23}}\cos(\phi_{23})+\sqrt{\Delta_{12}\Delta_{14}\Delta_{24}}\cos(\phi_{24})+\sqrt{\Delta_{13}\Delta_{14}\Delta_{34}}\cos(\phi_{34})+\sqrt{\Delta_{23}\Delta_{24}\Delta_{34}}\cos(\phi_{34}-\phi_{24})   \right]\\
    -2\left[ \sqrt{\Delta_{12}\Delta_{13}\Delta_{14}\Delta_{34}}\cos(\phi_{34})+\sqrt{\Delta_{12}\Delta_{14}\Delta_{23}\Delta_{34}}\cos(\phi_{34}-\phi_{23})+\sqrt{\Delta_{13}\Delta_{14}\Delta_{23}\Delta_{24}}\cos(\phi_{24}-\phi_{23})   \right]>0.
\end{align*}
\end{widetext}
It is clear that any set of $m$ states without null overlaps can be analyzed similarly, with the linear independence condition expressed in terms of invariants of degree 3. As discussed in the main text, more complex orthogonality conditions results in frame graphs that require higher-order invariants for a complete characterization of the PU-equivalence class.




\section{Maximizing $\Im(\Delta)$} \label{app:imaginarity}

In this Appendix we describe the maximum values we found for $\Im(\Delta)$ for $n$-state Bargmann invariants, and show that for a suitable choice of a tuple of $n$ states $\Im(\Delta) \to 1$ as $n \to \infty$. Take a tuple of $n$ pure, single-qubit states forming successive vertices of a regular spherical $n$-gon on the surface of the Bloch sphere. They can be chosen as 
\begin{equation}
\ket{\psi_k}=\cos(\theta/2)\ket{0}+\sin(\theta/2)e^{ 2\pi ik/n}\ket{1},  \label{eq:ngon}
\end{equation}
for  $k \in \{0,1,2,\dots,(n-1)\}$, for some choice of $\theta$. For $n=3$, it is easy to show that $\Im(\Delta)$ reaches its maximal value of $\Im(\Delta)=1/4$ for the spherical triangle vertices given by $\theta=\arccos{(1/\sqrt{3})}$. More generally, for the regular spherical $n$-gon vertices we have:
\begin{equation}
\Delta=[1+\sin^2(\theta/2) (e^{2\pi i/n}-1) ]^n \label{eq:deltan}
\end{equation}

For $n \ge 3$, the tuples of states described by Eq. (\ref{eq:ngon}) have the highest value of $\Im(\Delta)$ that we have found. For example, for $n=10$ we obtain $\Im(\Delta)\approx 0.6949$ for $\theta\approx 0.9972$. Using Eq. (\ref{eq:deltan}) it is easy to show that $\Im(\Delta) \to 1$ as $n \to \infty$, if we choose $\theta = \pi/3$.
\section{Basis-independent coherence witnesses} \label{app:cohe}

Here we look at the simplest scenarios for which coherence witnesses can be given in terms of Bargmann invariants. For $m=2$ states $\rho_1, \rho_2$, it is possible to guarantee coherence if we measure the overlap $\Delta_{12}$ together with each state's purity $\Delta_{ii}=Tr(\rho_i^2)$ - for example, $\Delta_{11}=\Delta_{22}=1$ means both states are pure, so that any $\Delta_{12}\neq 0,1$ will serve as a witness of coherence. The situation is more interesting when we only measure relational information between different states.

 Let us now analyze the scenario with $m=3$ states, where we measure 2- and 3-state Bargmann invariants $\Delta_{12}, \Delta_{13}, \Delta_{23}$, and $\Delta_{123}$. As discussed before, for simultaneously diagonal (i.e., coherence-free) states $\Delta_{123}$ is the probability that independent measurements of the reference observable on the 3 states give the same outcome (and similarly for $\Delta_{ij}$ and state pair $\rho_i, \rho_j$). So the first type of witness would result if the 3-state Bargmann invariant was complex, or real negative. Following \cite{GalvaoB20}, we find the polytope of allowed  4-tuples $(\Delta_{12}, \Delta_{13}, \Delta_{23}, \Delta_{123})$ by listing the allowed deterministic states in this space, and calculating their convex hull. For this scenario, the allowed deterministic states are given in Table I. 
\begin{table}[ht]
  \begin{center}
    \begin{tabular}{|c|c|c|c|}
    \hline
      \textbf{$\Delta_{12}$} & \textbf{$\Delta_{13}$} & \textbf{$\Delta_{23}$} &
      \textbf{$\Delta_{123}$}\\
      \hline \hline
     0 & 0 & 0 & 0\\ \hline
     0 & 0 & 1 & 0\\ \hline
     0 & 1 & 0 & 0\\ \hline
     1 & 0 & 0 & 0\\ \hline
     1 & 1 & 1 & 1\\ \hline
    \end{tabular} \label{poly3}
\end{center}
\caption{Each row is one of the possible deterministic assignments to the Bargmann invariants describing the 3-state scenario. The convex hull of the rows characterizes all allowed 4-tuples of invariants for simultaneously diagonal, i.e. coherence-free states.}
\label{table:and2}
\end{table}
The corresponding polytope facets, characterizing all 4-tuples of invariants consistent with coherence-free states, are given by inequalities:
\begin{align}
\Delta_{123}&\ge 0, \label{eq:posd1}\\
  \Delta_{123}&\le \Delta_{ij}, \text{for } ij = 12, 13, 23, \label{eq:posd2}\\
\Delta_{123} &\ge \frac{1}{2}\left( \Delta_{12}+\Delta_{13}+\Delta_{23}-1\right). \label{eq:posd3}
\end{align}
Any violation of ineqs.\ (\ref{eq:posd1}-\ref{eq:posd3}) is a witness of coherence, and these inequalities are strictly stronger to detect coherence than the imaginarity witnesses we have introduced, or previous basis-independent coherence witnesses based on overlaps \cite{GalvaoB20}, measured recently in  \cite{Giordani21}. We can see this with a simple example. Consider 3 single-qubit states equally spaced  on the $xz$ plane great circle of the Bloch sphere. These states have real amplitudes, so their coherence cannot be witnessed via imaginarity witnesses; it can be easily checked also that these equal-valued overlaps $r_{12}=r_{13}=r_{23}=1/4$ are compatible with coherence-free states \cite{GalvaoB20}. This set, however, results in $\Delta_{123}=-1/8$, hence witnesssing coherence due to violation of ineq.\ (\ref{eq:posd1}).


\bibliography{bosonsamplingrefs}

\end{document}